\documentclass[11pt]{article}
\pdfoutput=0
\usepackage{fullpage}
\usepackage{times}

\newif\ifstoc
\stoctrue 
\stocfalse

\usepackage[ruled]{algorithm2e}
\usepackage{named}
\renewcommand{\algorithmcfname}{ALGORITHM}
\SetAlFnt{\small}
\SetAlCapFnt{\small}
\SetAlCapNameFnt{\small}
\SetAlCapHSkip{0pt}
\IncMargin{-\parindent}

\usepackage{graphicx,color}
\usepackage{array,float}
\usepackage{url}
\usepackage[usenames,dvipsnames]{xcolor}
\usepackage{amstext,amssymb,amsmath}
\usepackage{amsthm} 
\usepackage{mathtools}
\usepackage{verbatim}
\usepackage{bm}
\usepackage{paralist}
\usepackage{mathrsfs}
\usepackage{ulem}
\usepackage{dsfont}
\usepackage{todonotes}
\usepackage[noend]{algorithmic}
\usepackage{pstricks, pst-tree, pst-node, subfigure}
\usepackage{enumitem}

\usepackage{geometry}

\usepackage{caption}

\ifstoc
\else
\fi

\ifstoc
\else
\fi

\newtheorem{lem}{Lemma}[section]
\newtheorem{thm}[lem]{Theorem}
\newtheorem{cor}[lem]{Corollary}

\newtheorem{defn}[lem]{Definition}

\renewcommand{\paragraph}[1]{\vspace{3pt}\noindent\textbf{#1}}

\newcommand{\D}{\mathcal{D}}

\newcommand{\ktnote}[1]{}

\newcommand{\ignore}[1]{}

\DeclareMathOperator*{\argmax}{{\sf argmax}}

\theoremstyle{remark} 

\newcommand{\items}{\mathcal{I}}
\newcommand{\numitems}{m}
\newcommand{\buyers}{\mathcal{B}}
\newcommand{\numbuyers}{n}
\newcommand{\xxi}{x^{(i)}}
\newcommand{\binaryitems}{\{0,1\}^{\numitems}}

\newcommand{\agents}{\mathcal{A}}
\newcommand{\goods}{\mathcal{C}}
\newcommand{\numagents}{N}
\newcommand{\numgoods}{M}
\newcommand{\wwi}{w^{(i)}}
\newcommand{\arrowdeb}{ Arrow-Debreu market }
\usepackage[T1]{fontenc}
\usepackage[utf8]{inputenc}
\usepackage{authblk}
\begin{document}

\markboth{Niazadeh and Wilkens}{Competitive Equilibria for Non-quasilinear Bidders in Combinatorial Auctions}
	
\title{Competitive Equilibria for Non-quasilinear Bidders\\ in Combinatorial Auctions}

\author[1]{Rad Niazadeh\thanks{rad@cs.cornell.edu}}
\author[2]{Christopher Wilkens\thanks{cwilkens@yahoo-inc.com}}
\affil[1]{Department of Computer Science, Cornell University}
\affil[2]{Yahoo Research Labs}

\renewcommand\Authands{ and }
\maketitle
\begin{abstract}
Quasiliearity is a ubiquitous and questionable assumption in the standard study of Walrasian equilibria. Quasilinearity implies that a buyer's value for goods purchased in a Walrasian equilibrium is always additive with goods purchased with unspent money. It is a particularly suspect assumption in combinatorial auctions, where buyers' complex preferences over goods would naturally extend beyond the items obtained in the Walrasian equilibrium.

We study Walrasian equilibria in combinatorial auctions when quasilinearity is not assumed. We show that existence can be reduced to an Arrow-Debreu style market with one divisible good and many indivisible goods, and that a ``fractional'' Walrasian equilibrium always exists. We also show that standard integral Walrasian equilibria are related to integral solutions of an induced configuration LP associated with a fractional Walrasian equilibrium, generalizing known results for both quasilinear and non-quasilnear settings.
\end{abstract}

\section{Introduction}
Money is inherently useless; it only holds value because of the promise that it can be used to buy something useful. Thus, an agent's utility for money will depend substantially on what she already has and the alternative ways that it can be spent. A student who saves money on her habitual cup of coffee can spend it on many things. On the other hand, a corporate event planner who is given a dedicated budget for drinks may not receive any benefit for discounted coffee; on the contrary, she cannot spend the money on herself, and her budget may get cut next time if she doesn't spend enough on the current event.

General market equilibria capture the ephemerality of money. Arrow and Debreu's exchange model is simple: agents have goods; they sell those goods, then buy what they want most. In this setting, money has no inherent value and is simply a lubricant facilitating exchange. This works because Arrow and Debreu capture the entire economy, so there is nothing outside the market on which to spend money.

In contrast, Walrasian (competitive) equilbria only capture a slice of the overall market. To do so, they must attribute utility to unspent money, since it can be spent elsewhere. In a Walrasian equilibrium model, agents arrive with money and only spend it if the received value outweighs the cost.

The default way to capture the implicit value of money is through a quasilinear utility function, i.e. a bidder's utility $u$ is the difference between her value $v$ and the price she pays $p$. Quasilinearity assumes that an agent's total value is additive in what she gets now and what she gets from an outside option, and that the amount of utility she gets from an outside option scales linearly with the amount of money she applies to it. Both are plausible modeling assumptions, but assuming that they are always true is only somewhat more defensible than assuming that a bidder always has an additive valuation over all items in the market.

We can illustrate one simple violation with our coffee example. Suppose our student will buy exactly one cup of coffee each day. If she doesn't buy coffee now, she will spend \$3 on coffee elsewhere, so her value for coffee now is \$3; however, if we give her {\em both} a cup of coffee {\em and} \$3, she will spend the extra \$3 on something completely different, like a movie ticket (she already has her cup of coffee). Whether this movie ticket gives her the same utility as a cup of coffee, or half the utility, or a quarter of the utility --- there is nothing in our market model to imply that her utility from an extra \$3 is in any way tied to her utility for a cup of coffee, except that it is plausibly {\em at most} her utility for a cup of coffee.\footnote{We know that she chose a cup of coffee over a movie ticket initially, so that implies her value for a cup of coffee is less than her value for a movie ticket. On the other hand, there might also be complementarities here if the student is unable to enjoy the movie without first having a cup of coffee...} In effect, the student's marginal utility for \$3 is completely different depending on whether or not she gets coffee now.

Relaxing quasilinearity for Walrasian equilbria is thus a fundamental question, particularly in combinatorial auction domains where bidders are assumed to have complex preferences over sets of items. This is the topic of our paper. Existing relaxations of quasilinearity focus on unit demand agents, and the literature is quite limited. The existence of Walrasian equilibria was first shown by Quinzii~\shortcite{quinzii1984core}. Later, Alaei et al.~\shortcite{alaei2011competitive} show that Walrasian equilibria exist using a direct approach and have the same kind of lattice structure as standard Walrasian equilibria for unit demand bidders. Maskin~\shortcite{maskin1987fair} studies a superficially different problem and adds a single divisible good (i.e. money) to a standard general market model with indivisible goods; his result is that market equilibria always exist. We study Walrasian equilibria with general utilities in an even more obvious setting: combinatorial auctions. In a combinatorial auction, bidders may have complex preferences over the multiple goods being sold. Thus it is natural that one should extend these complex preferences to items outside the Walrasian micromarket by allowing non-quasilinear preferences for money. Our main results establish conditions under which a Walrasian equilibrium exists.

\paragraph{Fractional Walrasian Equilibria and Market Equilibria}

Walrasian equilibria capture a microcosm of a larger market. To do so, they capture both an agent's utility for unspent money and goods' inherent indivisibility at small scales. As our prior discussion suggests, money is simply a proxy for the portion of the general market outside the goods available in the Walrasian equilibrium. It is therefore not surprising that the first step in our work constructs a reduction from a Walrasian equilibrium problem to a general Arrow-Debreu market with an extra good (money), and an extra agent (the selller). Together, the extra good and agent capture the market outside the Walrasian equilibrium. This market is special because all goods except money have supply 1 and are indivisible. Maskin~\cite{maskin1987fair} studies a similar market without adding the seller as an agent. Our first result shows that the equilibria are the same:
\begin{lem}[Informal]
A set of prices and allocations for a combinatorial auction are a Walrasian equilibrium if and only if they correspond to a market equilibrium of the associated special Arrow-Debreu market.
\end{lem}
The relationship with the special Arrow-Debreu market also lays the foundation for our first main result --- ``fractional'' Walrasian equilibria always exist:
\begin{thm}[Informal]
In a combinatorial auction setting, a fractional Walrasian equilibirium always exists, that is there exists a set of prices and, for each player, a distribution over sets of goods whose support is the collection of all demanded sets under the equilibrium prices, and market clears in expectation. 
\end{thm}
This follows by proving that an auction's associated general market model always has a fixed point, and that demand of a bidder at the fixed point can always be decomposed into a distribution over goods.

\paragraph{Configuration LPs and True Walrasian Equilibria}

To understand true (non-fractional) Walrasian equilibria, we must understand the combinatorial auction's {\em configuration linear program}. In a quasilinear setting, the configuration LP captures the way goods can be (fractionally) assigned to bidders; the LP's objective is the total value generated by the assignment. Walrasian equilibria here are known to be equivalent to integral optima of an auction's configuration LP.

Unfortunately, the configuration LP is not available without quasilinearity because a bidder's value is not well-defined. Instead, we introduce an {\em induced configuration LP} that is associated with a particular price vector $p^*$. This LP is constructed by fixing bidders' utilities at $p^*$ and assuming they are otherwise quasilinear. Integral optima are again related to Walrasian equilibria, but only if $p^*$ already supported a fractional Walrasian equilibrium:
\begin{thm}[Informal]
A Walrasian equilibrium for a combinatorial auction exists if and only if there is a price vector $p^*$ supporting a fractional Walrasian equilibrium for which the induced configuration LP has an integral optimum.
\end{thm}
This theorem has a couple of interesting corollaries. First, we can see how it relates to the results of Maskin~\shortcite{maskin1987fair}, Quinzii~\shortcite{quinzii1984core}, and Alaei et al.~\shortcite{alaei2011competitive}:
\begin{cor}
If the induced configuration LP is always integral, a combinatorial auction always has a Walrasian equilibrium.
\end{cor}
In unit-demand settings, the induced configuration LP is a matching LP at any fixed price $p^*$. Thus, it is integral and always has an integral optimum, and Walrasian equilibria will always exist. This generalizes the results of Maskin, Quinzii and Alaei et al. --- Maskin studied a sibling of our general market model and effectively showed that when bidders are unit demand, there is always an integral solution for every fixed point, while Alaei et al. directly show that Walrasian equilibria exist for unit demand settings.

Another simple but useful corollary happens when the configuration LP is independent of $p^*$:
\begin{cor}
If the induced configuration LP is independent of $p^*$, then a combinatorial auction has a Walrasian equilibrium if and only if the configuration LP has an integral optimum, regardless of any properties of $p^*$.
\end{cor}
The dependence on $p^*$ cancels when utilities are quasilinear, and it explains why we can simply talk about the configuration LP without talking about a specific set of prices $p^*$.

Together, these results build a picture of the equilibrium landscape outside quasilinearity --- Walrasian equilibria still exist at least in a fractional form, but testing existence in general is substantially more complicated. Existence is still related to a configuration LP, but that LP can only be defined once prices $p^*$ are in hand.

\paragraph{Related Work}
The objective of this paper is establishing existence characterization for competitive equilibria in combinatorial auction. The problem is closely related to the existing literature in economics and theoretical computer science from different directions. First, there is a prominent literature on competitive equilibrium in combinatorial auctions for the quasilinear setting. These papers characterize existence conditions for a Walrasian equilibrium and provide practical necessary conditions for a competitive equilibrium to exist in the quasilinear setting such as gross substitute, e.g. Gul and Stacchetti ~\shortcite
{gul1999walrasian}, Bikhchandani and Mamer~\shortcite{bikhchandani1997competitive}, Bevia, Quinzii, and and Silva\shortcite{bevia1999buying}, Murota and Tamura~\shortcite{murota2001computation}. People have also thought about using the properties of competitive equilibria in quasilinear settings when valuations are satisfying the gross substitute, such as lattice structure\cite{gul1999walrasian}, in order to design ascending auctions and identifying the connections between the well-known VCG mechanism~\cite{vickrey1961counterspeculation,groves1973incentives,clarke1971multipart} and Walrasian equilibria~\cite{kelso1982job,cramton2006combinatorial,nisan2007algorithmic}. Closely related is the literature on assignment games and core allocations, in which they tried to generalize the stable matching concept (either one-to-one matching, one-to-many matching, or many-to-many matching) to two-sided markets with indivisible goods and quasilinear utilities, e.g. Shapley and Shubik~\shortcite{shapley1971assignment} and Echenique et al.~\shortcite{echenique2004theory}.

The second direction that connects our work to the literature is the existing work on non-quasilinear utilities (or non-transferable utilities) and two-sided matching markets or general \arrowdeb~\cite{arrow1954existence} with only one divisible good and unit-demand agents. The problem formulation is as introduced by Demange and Gale~\shortcite{demange1985strategy}. The existence of competitive equilibria for non-quasilinear utilities was first proved by Quinzii~\shortcite{quinzii1984core}, Gale ~\shortcite{gale1984equilibrium}, Svensson~\shortcite{svensson1984competitive}, and later by Kaneko and Yamamoto~\shortcite{kaneko1986existence}. They showed if there is a single divisible good (say money) in an economy and if agents are unit-demand then there still exists a competitive equilibrium under certain reasonable (monotonicity) conditions. There is also the work of Maskin\shortcite{maskin1987fair} on fair allocation of indivisible goods with money, that provides a simpler proof for the existence of the equilibrium with indivisible goods and only one divisible item in the unit-demand case. More recently, there has been a work on two-sided matching markets with non-transferable utilities by Alaei et al.~\shortcite{alaei2011competitive} that followed a different combinatorial approach. They showed the existence of competitive equilibrium when utility functions are monotone, and generalized the lattice structure and properties associated with the minimum lattice point to a general non-quasilinear setting. 
\section{Settings and Notations}
We are looking at a \emph{combinatorial auction}, in which we have a set $\items$ of $\numitems$ items and a set $\buyers$ of  $\numbuyers$ buyers interested in these items. For every  $x\in\{{0,1\}}^\numitems$ and price $p\in \mathbb{R}_{+}$, let $u_i(x,p)$ be the utility of bidder $i$ if she gets bundle $X=\{j\in \items: x_j=1\}$ of items under the price $p$. We assume utilities are strictly decreasing and continuous with respect to the price $p$, and increasing with respect to $x_j$ for every item $j$. The competitive equilibrium (also known as \emph{Walrasian equilibrium}) can be defined in this setting as follows.
\begin{defn}
\label{def:WE}
A  Walrasian Equilibrium (WE) is a pair of allocation and prices $(\{\xxi\}_{i=1}^\numbuyers,\{p_j\}_{j=1}^{\numitems})$ that satisfies the following conditions: 
\begin{itemize}
\item $\forall i\in \buyers$ and $j\in\items:~\xxi\in \binaryitems$ and $p_j\in \mathbb{R}_+$.
\item \textbf{[Feasibility]} $\forall j\in \items:~\sum_{i=1}^{\numbuyers}\xxi_{j}\leq 1$.
\item \textbf{[Satisfaction]} $\forall i\in \buyers:~\xxi\in\underset{x'\in \binaryitems}{\argmax}u_i(x',\sum_{j=1}^{\numitems} p_jx'_j)$,
\item \textbf{[Market clearance]} $\forall j\in\items$, if $p_j>0$ then $\sum_{i=1}^n\xxi_j=1$.
\end{itemize}
\end{defn}
Besides WE, we also need to define a fractional equilibrium, in which each buyer has a distribution over bundles of items. However, such an allocation is only feasible in expectation, meaning that each item gets allocated with probability less than or equal to $1$. Note that such an equilibrium cannot be realized in reality and it is just a solution concept that will shed insight on the structure of WE, as we show later in this paper.  More precisely, we have the following definition:
\begin{defn} 
\label{def:fracWE}
In a combinatorial auction $( \items,\buyers,\{u_i(.)\}_{i\in\buyers})$, a \emph{fractional} WE is defined to be a pair of allocation and prices $(\{x_{i,S}\},\{p_j\})$ such that
\begin{itemize}
\item $\forall (i,S)\in \buyers\times 2^{\items}$ and $j\in\items:~~$ $x_{i,S}\in\mathbb{R}_{+}$ and $p_j\in\mathbb{R}_{+}$.
\item\textbf{[Feasibility]}~$\forall j\in\items:\displaystyle\sum_{i\in \buyers, S:j\in S}x_{i,S}\leq 1$,$~~~~~\forall i\in\buyers: \displaystyle\sum_{S\subseteq\items}x_{i,S}=1$.
\item\textbf{[Satisfaction]}~$\forall i\in \buyers$, if $x_{i,S}>0$ then $\mathbf{1}_S\in\underset{x'\in \binaryitems}{\argmax}u_i(x',\sum_{j=1}^{\numitems} p_jx'_j)$.
\item \textbf{[Market clearance]}~$\forall j\in \items$, if $p_j>0$ then $\displaystyle \sum_{i\in \buyers, S:j\in S}x_{i,S}=1$.
\end{itemize}
\end{defn}

In this paper, we also talk about a special case of \emph{Arrow-Debreu markets} in which all commodities except one are indivisible. In such markets, there is a set $\agents$ of $\numagents$ of agents in the market who are interested in trading a set $\goods$ of $\numgoods$ commodities. We assume all commodities are indivisible except the last commodity $j=\numgoods$ (we sometimes call this commodity `money'). Each agent $i$ will bring an endowment $\wwi \in\mathbb{R}_+^{\numgoods}$ of commodities to the market to trade. Agent $i$ gets a utility of $\tilde{u}_{i}(x)$ for an allocation  $x \in \{0,1\}^{\numgoods-1}\times\mathbb{R}_+$ of commodities. Moreover, we assume $\tilde{u}_i(x)$ is strictly increasing with respect to $x_j$ for all $j\in\goods$ and continuous with respect to allocation of money, i.e. $x_{\numgoods}$. We next define the \emph{general market equilibrium} for such a market. 
\begin{defn}
A  General Market Equilibrium (GME) is a pair of allocation and prices $(\{\xxi\}_{i=1}^\numagents,\{p_j\}_{j=1}^{\numgoods})$ that satisfies the following conditions:
\begin{itemize}
\item $\forall i\in \agents$ and $j\in\goods/\{\numgoods\}:~\xxi_j\in \{0,1\} $ and $p_j\in\mathbb{R}_+$. 
\item $\forall i \in \agents:~\xxi_{\numgoods}\in\mathbb{R}_+$ and $P_M\in \mathbb{R}_+$.
\item \textbf{[Satisfaction]} $\forall i\in \agents:\xxi\in \underset{{j\neq \numgoods}:x'_j\in \{0,1\}, x'_{\numgoods}\in\mathbb{R}_+}{\argmax}\tilde{u}_i(x')~\textrm{s.t}.~ \sum_{j=1}^{\numgoods}x'_jp_j\leq \sum_{j=1}^{\numgoods}\wwi_jp_j$.
\item \textbf{[Market clearance]}~$\forall j\in \goods$, if $p_j>0$ then $\sum_{i=1}^\numagents \xxi_j=\sum_{i=1}^{\numagents}\wwi_j$.
\end{itemize}
\end{defn}
To define more notations for an \arrowdeb with only one divisible good, let $D_i(\{p_j\}_{j\in \goods})$ be the collection of feasible allocation of commodities to agent $i$, such that each maximizes utility of agent $i$ under prices $\{p_j\}_{j\in \goods}$ and they satisfy the budget constraint of agent $i$. In other words:
\begin{equation}
D_i(\{p_j\}_{j\in \goods})\triangleq \underset{{j\neq \numgoods}:x'_j\in \{0,1\}, x'_{\numgoods}\in\mathbb{R}_+}{\argmax}\tilde{u}_i(x')~\textrm{s.t.}~ \sum_{j=1}^{\numgoods}x'_jp_j\leq \sum_{j=1}^{\numgoods}\wwi_jp_j
\end{equation}
Let $\bar{D}_i({\{p_j\}_{j\in \goods}})$ be all vectors in $D_i(\{p_j\}_{j\in \goods})$ when we delete the allocation of the divisible good, i.e. last coordinate, from all vectors. Define total demand to be $D({\{p_j\}_{j\in \goods}})\triangleq \sum_{i\in \agents}D_i(\{p_j\}_{j\in \goods})$ and the total demand for items to be $\bar{D}({\{p_j\}_{j\in \goods}})\triangleq\sum_{i\in \agents}\bar{D}_i(\{p_j\}_{j\in \goods})$, where summations are Minkowski summations of sets. Moreover, let $\tilde{D}_i=\textrm{Conv}(D_i)$ and $\tilde{D}=\textrm{Conv}(D)$ where $\textrm{Conv}(.)$ is the convex hull of its argument. Clearly, all these sets are finite (because utility is strictly increasing in money and hence given an allocation of indivisible items the allocation of money will be the unique number that fills the budget slack) and hence convex hulls are well defined. Also, from the definition of convex hull and Minkowski summation, $\tilde{D}=\sum_{i\in \agents}\tilde{D}_i$.

\section{Reduction from combinatorial auction to \arrowdeb}
We start by defining a general \arrowdeb with one divisible good.
\begin{defn}
\label{def:red_market}
Given a combinatorial auction $( \items,\buyers,\{u_i(.)\}_{i\in\buyers})$,  its \emph{corresponding \arrowdeb} $(\goods,\agents,\{\tilde{u}_i(.)\}_{i\in \agents}, \{\wwi\}_{i\in\agents})$ is the following:
\begin{itemize}
\item There are $\numgoods=\numitems+1$ commodities, where the first $\numitems$ indivisible commodities in $\goods$ are items in $\items$, and the last divisible commodity is a special commodity called `money'. 
\item There are $\numagents=\numbuyers+1$ agents:
\begin{itemize}
\item For every $ i\in [\numbuyers]$ we have $\tilde{u}_i(x,y)=u_i(x,\frac{Z}{\numbuyers}-y)$ for every $x\in \{0,1\}^{\numitems}, y\in \mathbb{R}_+$, where $Z$ is large enough such that  $Z>\sum_{i=1}^n u_i(\mathbf{1},0)$,
\item The last agent is a special agent called the `seller' and her utility is computed as $\tilde{u}_{\numbuyers+1}(x,y)=y$ for every $x\in \{0,1\}^{\numitems}, y\in \mathbb{R}_+$.
\end{itemize}
\item  For every $i \in [\numbuyers]$, endowment of agent $i$ is $\wwi=(0,0,\ldots,0,\frac{Z}{\numbuyers})$. For the seller, $w_{\numbuyers+1}=(1,1,\ldots,1,0)$.
\end{itemize}
\end{defn}
We now have the following lemma, which basically shows that the correspondence described in Definition~\ref{def:red_market} preserves the equilibrium.
\begin{lem}
\label{lem:reduction}
A pair $(\{\xxi\}_{i=1}^\numbuyers,\{p_j\}_{j=1}^{\numitems})$ is a WE for the combinatorial auction $( \items,\buyers,\{u_i(.)\}_{i\in\buyers})$ if and only if there exists a GME $(\{\tilde{x}^{(i)}\}_{i=1}^\numagents,\{\tilde{p}_j\}_{j=1}^{\numgoods})$ for its corresponding \arrowdeb $(\goods,\agents,\{\tilde{u}_i(.)\}_{i\in \agents}, \{\wwi\}_{i\in\agents})$ as in Definition~\ref{def:red_market} such that $\forall (i,j)\in [\numagents-1]\times [\numgoods-1]:\tilde{x}^{(i)}_j=\xxi_j$, $\tilde{x}^{(\numagents)}_{\numgoods}=\sum_{j=1}^{\numitems}p_j$, $\forall j\in [\numgoods-1]:\frac{\tilde{p}_j}{\tilde{p}_\numgoods}=p_j$ and $\tilde{p}_{\numgoods}>0$. 
\end{lem}
\begin{proof}
We first show that given a WE $(\{\xxi\}_{i=1}^\numbuyers,\{p_j\}_{j=1}^{\numitems})$ for the combinatorial auction, we can generate a GME for its corresponding \arrowdeb that satisfies desired conditions. For every $j\in [\numgoods-1]$, let $\tilde{p}_j=p_j$ and let $\tilde{p}_\numgoods=1$.  For all $(i,j)\in [\numagents-1]\times [\numgoods-1]$ let $\tilde{x}^{(i)}_j=\xxi_j$. Let the seller collect all trade money, i.e. $\tilde{x}^{(\numagents)}_{\numgoods}=\sum_{j=1}^{\numitems}p_j$ and $\tilde{x}^{(\numagents)}_{j}=0, j\in[\numgoods-1]$. Finally, let $\tilde{x}^{(i)}_{\numgoods}=\frac{Z}{\numbuyers}-\sum_{j=1}^{\numgoods-1}\xxi_jp_j~,~i\in [\numagents-1]$. 
It is easy to show that market clears, because for $j\in[\numgoods-1]$, if $\tilde{p}_j>0$, then $p_j>0$ and hence the item $j$ gets allocated in the combinatorial auction. Consequently, 
\begin{equation}
\sum_{i=1}^{\numagents}\tilde{x}^{(i)}_j=\sum_{i=1}^{\numbuyers}\xxi_j=1=\sum_{i=1}^{\numagents} \wwi_j
\end{equation}
 Moreover, for money item (for which $\tilde{p}_{\numgoods}>0$), we have
\begin{equation}
\sum_{i=1}^{\numagents}\tilde{x}^{(i)}_{\numgoods}=
Z-\sum_{i=1}^{\numagents-1}\sum_{j=1}^{\numitems}\tilde{x}^{(i)}_j\tilde{p}_j+\sum_{j=1}^{\numitems}p_j
=Z+\sum_{j=1}^{\numitems}p_j(1-\sum_{i=1}^{\numbuyers}\xxi_j)\overset{(1)}{=}\sum_{i=1}^{N}\wwi_{\numgoods}
\end{equation}
where (1) is true, because if $p_j>0$, $\sum_{i=1}^{\numbuyers}\xxi_j=1$, and $\sum_{i=1}^{N}\wwi_{\numgoods}=Z$. To show the satisfaction condition, it is obvious that the seller gets its optimal allocation under budget constraints (it sells all the items and gets the money for it). For an agent $i\in[\numagents-1]$, fix any feasible allocation $x'$ of commodities to this agent. Now, we have:
\begin{align}
\tilde{u}_i(\tilde{x}^{(i)})&=u_i(\{\xxi_j\}_{j\in[\numitems]},\frac{Z}{\numbuyers}-\tilde{x}^{(i)}_{\numgoods})=u_i(\{\xxi_j\}_{j\in[\numitems]},\sum_{j=1}^{\numitems}\xxi_jp_j)\geq u_i(\{x'_j\}_{j\in[\numitems]},\sum_{j=1}^{\numitems}x'_jp_j)\nonumber\\
&= \tilde{u}_i(\{x'_j\}_{j\in[\numitems]},\frac{Z}{\numbuyers}-\sum_{j=1}^{\numitems}x'_jp_j)\overset{(1)}{\geq} \tilde{u}_i(\{x'_j\}_{j\in[\numitems]},x'_{\numgoods}),
\end{align}
where (1) holds, because utilities in the market are increasing with respect to money commodity and for a feasible allocation $x'$, $\sum_{j=1}^{\numitems}x'_j p_j+x'_{\numgoods}=\sum_{j=1}^{\numgoods-1}x'_j\tilde{p}_j+x'_{\numgoods}\leq\frac{Z}{\numbuyers}$. It is also true that $\tilde{x}^{(i)}$ meets the budget constraint, i.e. $\sum_{j=1}^{\numgoods-1}\tilde{x}^{(i)}_j\tilde{p}_j+\tilde{x}^{(i)}_{\numgoods}\leq\frac{Z}{\numbuyers}$, which completes the proof. 

To prove the other direction, suppose a GME $(\{\tilde{x}^{(i)}\}_{i=1}^\numagents,\{\tilde{p}_j\}_{j=1}^{\numgoods})$ is given and it satisfies the conditions in the statement of the lemma. Now $\forall (i,j)\in [\numbuyers]\times [\numitems]$ let $\xxi_j=\tilde{x}^{(i)}_j$. Moreover, let  $p_j=\frac{\tilde{p}_j}{\tilde{p}_\numgoods}, j\in[\numitems]$. Clearly the allocation $\{\xxi\}_{i\in [\numbuyers]}$ is a feasible allocation and clears the market due to the fact that $\{\tilde{x}^{(i)}\}_{i\in [\numagents]}$ clears the \arrowdeb and $\tilde{x}^{(\numagents)}_j=0$ in any GME (because seller's utility is strictly increasing with money and she gets zero utility by receiving any commodity $j\in[\numgoods-1]$). To see satisfaction, , fix an allocation of items to bidder $i$, e.g. $x'$, in the combinatorial auction. Now we have
\begin{align}
u_i(\{\xxi_j\}_{j\in[\numitems]}, \sum_{j=1}^{\numitems}\xxi_j p_j)&=\tilde{u}_i(\{\tilde{x}^{(i)}_j\}_{j\in[\numitems]}, \frac{Z}{\numbuyers}-
\sum_{j=1}^{\numitems}\tilde{x}^{(i)}_j \frac{\tilde{p}_j}{\tilde{p}_{\numgoods}})\overset{(1)}{\geq}\tilde{u}_i(\{\tilde{x}^{(i)}_j\}_{j\in[\numitems]},\tilde{x}^{(i)}_{\numgoods})\nonumber\\
&\overset{(2)}{\geq}\tilde{u}_i(\{x'_j\}_{j\in[\numitems]},\frac{Z}{\numbuyers}-
\sum_{j=1}^{\numitems}x'_j \frac{\tilde{p}_j}{\tilde{p}_{\numgoods}})=u_i(\{x'_j\}_{j\in[\numitems]}, \sum_{j=1}^{\numitems}x'_j p_j)
\end{align}
where (1) holds because $\sum_{j=1}^{\numgoods-1}\tilde{x}^{(i)}_j\tilde{p_j}+
\tilde{x}^{(i)}_{\numgoods}\tilde{p}_{\numgoods}\leq\tilde{p}_{\numgoods}\frac{Z}{n}$, and (2) holds because $(\{x'_j\}_{j\in[\numitems]},\frac{Z}{\numbuyers}-
\sum_{j=1}^{\numitems}x'_j \frac{\tilde{p}_j}{\tilde{p}_{\numgoods}})$ meets the budget constraint of bidder $i$ in \arrowdeb, or more concretely:
\begin{equation}
\sum_{j=1}^{\numgoods-1}x'_j\tilde{p}_j+
(\frac{Z}{\numbuyers}-
\sum_{j=1}^{\numitems}x'_j \frac{\tilde{p}_j}{\tilde{p}_{\numgoods}})\tilde{p}_{\numgoods}=\tilde{p}_{\numgoods}\frac{Z}{n}=\sum_{j=1}^{\numgoods}\wwi_j\tilde{p}_j
\end{equation}
\end{proof}
\section{A Generalization of configuration LP}
In this section, we start exploring the connections between welfare maximization and WE for non-quasilinear utilities in combinatorial auctions. In the quasilinear world, where $u_i(x,p)=v_i(x)-p$, the following connections are known:
\begin{itemize}
\item The combinatorial auction configuration LP, i.e. the following linear program 
\begin{equation*}
\label{eq:conflpq}
\begin{array}{ll@{}ll}
\text{maximize} & \displaystyle \sum_{i\in\buyers,S\subseteq \items}x_{i,S}v_i(\mathbf{1}_S)\\
\text{subject to}& ~~~\displaystyle\sum_{S\subseteq \items}x_{i,S}\leq 1, &~~~~i\in\buyers.\\
               &~~~\displaystyle\sum_{i\in \buyers}\sum_{S\subseteq \items:j\in S}x_{i,S}\leq 1,&~~~~j\in\items.\\
         &~~~x_{i,S}\geq 0,&~~~~i\in\buyers , S\subseteq \items.  
\end{array}
\end{equation*}
that characterizes maximum welfare allocations, has an integral optimal solution if and only if WE exists. 
\item A vector of prices is a WE price vector if it forms an optimal solution to the dual of the configuration LP. Moreover, if a dual solution can be supported by an integral feasible primal, then it is a WE price vector.
\end{itemize}
The question we address here is how can one generalize these concepts to the case of non-quasilinear utilities, in the hope that they shed some insights on existence and structural properties of WE for non-quasilinear utilities. To this end, we first define an \emph{equivalent quasilinear value function} for each bidder, which will act similar to the value function in the quasilinear configuration LP. 
\begin{defn}
Fix a vector of prices $p^*$. For bidder $i$ with utility function $u_i(x,p)$, the equivalent quasilinear value function at price vector $p^*$ is defined as
\begin{equation}
v^{(p^*)}_i(x)\triangleq u_i(x,\sum_{j=1}^{\numitems}x_jp^*_j)+\sum_{j=1}^{\numitems}p^*_j
\end{equation}
\end{defn}

As it can be seen from the definition, the equivalent quasi-linear value function, together with prices $p^*$, will generate the same utility as the original utility function, if we assume quasi-linearity. Given the definition of an equivalent quasilinear value function for each bidder at a fixed price vector $p^*$, here is a natural generalization to the configuration LP. The program maximizes welfare with respect to the equivalent quasilinear value function.
\begin{defn}[Induced configuration LP at price $p^*$]
\label{def:confLP}
Fixing a price vector $p^*$, the \emph{induced configuration LP at price $p^*$} is defined as the following linear program with variables $\{x_{i,S}\}_{i\in\buyers,S\subseteq\items}$ (allocation):
\begin{equation*}
\label{eq:conflp}
\begin{array}{ll@{}ll}
\text{maximize} & \displaystyle \sum_{i\in\buyers,S\subseteq \items}x_{i,S}v^{(p^*)}_i(\mathbf{1}_S)\\
\text{subject to}& ~~~\displaystyle\sum_{S\subseteq \items}x_{i,S}\leq 1, &~~~~i\in\buyers.\\
               &~~~\displaystyle\sum_{i\in \buyers}\sum_{S\subseteq \items:j\in S}x_{i,S}\leq 1,&~~~~j\in\items.\\
         &~~~x_{i,S}\geq 0,&~~~~i\in\buyers , S\subseteq \items.  
\end{array}
\end{equation*}
\end{defn}
Similar to the quasilinear utilities, one can look at the dual program of the linear program in Definition~\ref{def:confLP} which sheds more insights on the structure of the WE, as we show later in this paper.
\begin{defn}[Dual induced configuration LP at price $p^*$]
\label{def:dualconfLP}
Fixing a price vector $p^*$, the dual of the induced configuration LP in Definition~\ref{def:confLP} is the following linear program with variable $\{u_i\}_{i\in\buyers}$(utilities) and $\{p_j\}_{j\in \items}$(prices).
\begin{equation*}
\label{eq:dualconflp}
\begin{array}{ll@{}ll}
\text{minimize} & \displaystyle \sum_{i\in\buyers}u_i+\sum_{j\in\items}p_j\\
\text{subject to}& ~~~\displaystyle \sum_{j\in S}p_j+u_i\geq v^{(p^*)}_i(\mathbf{1}_S) , &~~~~i\in\buyers, S\subseteq \items.\\
         &~~~u_i\geq 0, p_j\geq 0,&~~~~i\in\buyers , j\in \items.  
\end{array}
\end{equation*}
\end{defn}
In the next section we show how the linear programs in Definitions~\ref{def:confLP} and \ref{def:dualconfLP} are related to the existence of WE in non-quasilinear settings. 

\section{main results and their applications}
Our main result is proving the existence of WE under necessary and sufficient structural conditions, and bridging the gap between the concept of WE and configuration LP for non-quasilinear utilities. More accurately, we show the induced configuration LP in Definition~\ref{def:confLP} is strong enough to provide us with necessary and sufficient conditions for the existence of equilibrium, however we have to look at this LP when $p^*$ is also an equilibrium price vector. Using the reduction in Section~\ref{def:red_market} to general markets, we show such item prices always exist. Moreover, it turned out that by using the primal-dual LP machinery one can show such prices will get supported by an integral allocation to form a WE if and only if the corresponding induced configuration LP has an integral optimal solution.

\begin{defn}
\label{def:marketmap}
Fix a combinatorial auction $( \items,\buyers,\{u_i(.)\}_{i\in\buyers})$ and consider its corresponding \arrowdeb $(\goods,\agents,\{\tilde{u}_i(.)\}_{i\in \agents}, \{\wwi\}_{i\in\agents})$, as defined in Section~\ref{def:red_market}. The \emph{market correspondence} $\phi$ is defined as follows.
\begin{equation}
\forall (p,d)\in \mathbb{R}_{+}^{\numgoods}\times\mathbb{R}_{+}^{\numgoods}: \phi(p,d)=(\tilde{D}(p),\mathcal{F}(d)),
\end{equation}
where $\mathcal{F}(d)\triangleq \underset{\hat{p}\in \mathbb{R}_{+}^{\numgoods}}{\textrm{argmax}}~ \hat{p}.(d-(1,1,\ldots,1,Z))$ and $\tilde{D}(p)$ is the convex hull of total demand set $D(p)$. We say a point $(p,d)$ is a \emph{fixed point} of the market correspondence if
\begin{equation}
(p,d)\in \phi(p,d)=(\tilde{D}(p)),\mathcal{F}(d)).
\end{equation}
\end{defn}

As we will show later, the fixed point of market correspondence $\phi$ defined in Definition~\ref{def:marketmap}  always exists, under monotonicity assumptions on utility functions. This fixed point is essentially giving us an equilibrium price vector that can be supported by a \emph{fractional} allocation of items to buyers in a way that it produces an envy-free market clearing outcome ( i.e. an outcome that everyone gets an optimal allocation under prices and market clears). However, we expect integral allocations in a WE of the combinatorial auction. To address this we utilize the induced configuration LP defined in Definition~\ref{def:confLP} and its dual in Definition~\ref{def:dualconfLP} at the fixed point price to see if the supporting fractional allocation can basically be decomposed into integral allocations. This helps us to find a structural characterization for WE. Putting all the pieces together, we get two main results.

\begin{thm} 
\label{thm:fracWE}
Given a combinatorial auction $( \items,\buyers,\{u_i(.)\}_{i\in\buyers})$  in which for every buyer $i$ the utility function $u_i(x,p)$ is increasing with respect to allocation of items, and strictly decreasing and continuous with respect to the money, a \emph{fractional} WE (Definition~\ref{def:fracWE}) always exists.
\end{thm}

\begin{thm}
\label{thm:intWE}
Given a combinatorial auction $( \items,\buyers,\{u_i(.)\}_{i\in\buyers})$ that satisfies conditions in Theorem~\ref{thm:fracWE}, and its corresponding \arrowdeb $(\goods,\agents,\{\tilde{u}_i(.)\}_{i\in \agents}, \{\wwi\}_{i\in\agents})$ as in Definition~\ref{def:red_market}, a pair of prices and allocation $(\{p_j\}_{j\in \items},\{\xxi\}_{i\in\buyers})$ is a WE if and only if:
\begin{itemize}
\item $\exists~\tilde{p}\in \mathbb{R}_+^{\numgoods}$ and $\tilde{d}\in\mathbb{R}_+^{\numgoods}$ s.t. $\tilde{p}_{\numgoods}>0$, $\tilde{d}_{\numgoods}=Z$, $j\in [\numitems]:p_j=\frac{\tilde{p}_j}{\tilde{p}_{\numgoods}}$ and $(\tilde{p},\tilde{d})$ is a fixed point of $\phi$. 
\item $\{x_{i,S}\}$ is an optimal integral solution for the induced configuration LP at prices $p^*=p$, where $\forall i\in\buyers, S\subseteq \items: x_{i,S}=1\iff j\in S:\xxi_j=1$.
\end{itemize}
\end{thm}

While our main results in Theorems~\ref{thm:fracWE} and Theorem~\ref{thm:intWE} characterize structural necessary and sufficient conditions for the existence of competitive equilibria, there are simple corollaries of these theorems that are of interest. The first corollary, whose proof can be seen from the proofs of Theorem~\ref{thm:fracWE} and \ref{thm:intWE} in Section~\ref{sec:proofs}, states the relationship between the dual of induced configuration LP at some price $p^*$ and prices in a competitive equilibrium (either fractional or integral). 
\begin{cor}
\label{cor:dual}
For a combinatorial auction  $( \items,\buyers,\{u_i(.)\}_{i\in\buyers})$ and its corresponding \arrowdeb $(\goods,\agents,\{\tilde{u}_i(.)\}_{i\in \agents}, \{\wwi\}_{i\in\agents})$, these statements are equivalent:
\begin{itemize}
\item Price vector $p$ is a fractional WE price vector, as in Definition~\ref{def:fracWE}.
\item There exists a price vector $\tilde{p}\in\mathbb{R}_+^\numgoods$ s.t. $(\tilde{p},\tilde{d})$ is a fixed point of the market correspondence $\phi$ and $p_j=\frac{\tilde{p}_j}{\tilde{p}_{\numgoods}}$.
\item Let $u_i=\underset{S\subseteq\items}\max ~v^{(p)}_i(\mathbf{1}_S)-\sum_{j\in S}p_j$.  Then $(u,p)$ is an optimal solution to the dual of induced configuration LP at price $p$.
\end{itemize}
\end{cor}

Another corollary of our main result is the connection between fractional WE, as in Defintion~\ref{def:fracWE}, and true (integral) WE, as in Definition~\ref{def:WE}. In fact, by directly applying Corollary~\ref{cor:dual} to Theorem~\ref{thm:intWE} one can restate Theorem~\ref{thm:intWE} through the following corollary, which bypasses the relationship to markets and reveals the relationship between fractional and integral WE.
\begin{cor}
\label{cor:alter}
Given a combinatorial auction $( \items,\buyers,\{u_i(.)\}_{i\in\buyers})$ that satisfies conditions in Theorem~\ref{thm:fracWE}, a pair of prices and allocation $(\{p_j\}_{j\in \items},\{\xxi\}_{i\in\buyers})$ is a WE if and only if:
\begin{itemize}
\item $p$ is a price vector of a fractional WE, as in Definition~\ref{def:fracWE}.
\item $\{x_{i,S}\}$ is an optimal integral solution for the induced configuration LP at prices $p^*=p$, where $\forall i\in\buyers, S\subseteq \items: x_{i,S}=1\iff j\in S:\xxi_j=1$.
\end{itemize}
\end{cor}

The next corollary of our results is the existence of competitive equilibrium for the special case of unit-demand bidders (or matching markets). In this case, each buyer is interested in at most one item and  a feasible allocation is an integral matching. Our result, somehow surprisingly, will give a simple proof for the existence of WE in this setting, which has been observed and proved in the literature  first by Quinzii~\shortcite{quinzii1984core}, and later by Alaei et.~\cite{alaei2011competitive}  and Maskin~\cite{maskin1987fair} - Maskin and Quinzii studied the matching markets with one divisible goods and showed the existence of an envy-free outcome in such a market, while Alaei et al. directly showed that competitive equilibrium exists by a combinatorial proof. 
\begin{cor}
\cite{alaei2011competitive,quinzii1984core,maskin1987fair} In a special case of unit-demand bidders, if the utilities are increasing with respect to the allocation and strictly decreasing and continuous with respect to money, competitive equilibrium always exists.
\end{cor}
\begin{proof}
Pick any price vector $p^*$ and look at the induced configuration LP at this price. Interestingly, the feasible polytope of such LP is the matching polytope. We know matching polytope is integral~\cite{schrijver1983short}, and hence there always exists an integral optimal solution to the induced configuration LP at any price vector $p^*$. As we show later in the proof of Theorem~\ref{thm:fracWE}, there always exists a fractional WE price vector  $p$ under the conditions in the statement of the corollary. Now, induced configuration LP at price $p$ has an optimal integral primal solution $x$ that supports any optimal solution $(\hat{p},\hat{u})$ of the dual program ( meaning that for each bidders $i$ the item she gets in $x$ is her preferred item under prices $\hat{p}$, she gets a utility $\hat{u}_i$ and also market clears ). According to Corollary~\ref{cor:dual}, $p$ is also an optimal dual solution for the induced configuration LP at price $p$, and hence $x$ supports $p$. So $(x,p)$ forms a WE for the matching market.
\end{proof}

The last corollary of our result is a simple proof for the classic result of Gul and Stacchetti~\cite{gul1999walrasian}, in which they demonstrate the relationship between competitive equilibria and configuration LP in the case of quasilinear utilities (which is the case when $u_i(x,p)=v_i(x)-p$ for all $i\in\buyers$, where $v_i(\mathbf{1}_S)$ denotes the value of bidder $i$ for bundle $S$). In fact, in the special case of quasilinear, the induced configuration LP at any price $p^*$ will not be a function of $p^*$, as $v_i^{(p^*)}(\mathbf{1}_S)=u_i(\mathbf{1}_S,\sum_{j\in S}p^*_j)+\sum_{j\in S}p^*_j=v_i(\mathbf{1}_S)$. We therefore have the following corollary.
\begin{cor}\cite{gul1999walrasian} Given a combinatorial auction with quasilinear utilities, a WE exists if and only if the configuration LP has an integral optimal solution. 
\end{cor}
\begin{proof}
Mixing Theorem~\ref{thm:fracWE} and Corollary~\ref{cor:alter}, we conclude there always exist a price vector $p$ such that it is a factional WE price vector and together with $\{u_i\}$,where $u_i=\underset{S\subseteq\items}\max ~v_i(\mathbf{1}_S)-\sum_{j\in S}p_j$, forms an optimal solution to the dual of configuration LP. So, using Theorem~\ref{thm:intWE} a WE exists if and only if there exists an integral optimal solution for the configuration LP.
\end{proof}

\section{Proof of the main results}
\label{sec:proofs}
\subsection{Proof of Theorem~\ref{thm:fracWE}}

We begin by looking at the market correspondence $\phi$ (Definition~\ref{def:marketmap}). We have the following lemma, whose proof is basically by Kakutani's fixed point theorem~\cite{kakutani1941generalization}. Checking the conditions of this theorem is technical and we omit the details for the sake of brevity. We assert that the proof is similar to the fixed-point proof in \cite{arrow1954existence} or \cite{maskin1987fair} with a minor modification. 
\begin{lem}
\label{lem:fixed}
If for all $i\in \buyers$, $u_i(x,p)$ satisfies the following conditions:
\begin{itemize}
\item  Continuous with respect to $p$,
\item Increasing with respect to $x_j:j\in\items$, 
\item Strictly decreasing with respect to $p$, 
\end{itemize}
then the market correspondence $\phi$ will have a fixed point.
\end{lem}
\begin{proof}
By the conditions in the statement of the lemma, the correspondence $\phi$ satisfies all the hypotheses of the Kakutani's fixed point theorem in \cite{kakutani1941generalization}, and it therefore has a fixed point.
\end{proof}
Due to Lemma~\ref{lem:fixed}, $\phi$ has a fixed point $(p^*,d^*)$, where $d^*\in \tilde{D}(p^*)=\sum_{i\in\agents}\tilde{D}_i(p^*)$ and
 \begin{equation*}
 p^*\in \mathcal{F}(d^*)=\underset{\hat{p}\in \mathbb{R}_{+}^{\numgoods}}{\textrm{argmax}}~ \hat{p}.(d^*-(1,1,\ldots,1,Z)).
\end{equation*}
 
 By the definition of $\tilde{D}(p^*)$ there exits $\{\tilde{d}^{(i)}\}_{i\in\agents}$ such that $\tilde{d}^{(i)}\in\tilde{D}_i(p^*)=\textrm{Conv}(D_i(p^*))$ and $d^*=\sum_{i\in\agents}\tilde{d}^{(i)}$. For each $i$, $\tilde{d}^{(i)}$ is a convex combination of vectors of the form $[\mathbf{1}_S,\alpha_S],$  each $[\mathbf{1}_S,\alpha_S]\in D_i(p^*)$ and therefor maximizing the utility of agent $i$ subject to the budget constraint under prices $p^*$ in the corresponding \arrowdeb. For each agent $i$, let $\{x_{i,S}\}$ be the coefficients of this convex combination. As a result, for each agent $ i\in[\numbuyers]$ the followings hold.
 \begin{align}
\label{eq:fracWEfeas}\sum_{S\subseteq[\numgoods-1]}&{x_{i,S}}=1,\\
\forall j\in[\numitems]:~~\sum_{S:j\in S}&x_{i,S}=\tilde{d}^{(i)}_j.
 \end{align}
 Now, because 1) at every point in $D_i(p^*)$ each agent satisfies her budget constraint under prices $p^*$,  2) each $\tilde{d}^{(i)}$ is a convex combination of such points, and 3) $d^*=\sum_{i\in\agents} \tilde{d}^{(i)}$, the following holds.
 \begin{equation*}
 p^*.d^*=\sum_{j\in[\numitems+1]}d^*_jp^*_j=\sum_{i\in[\numbuyers+1]}\sum_{j\in [\numitems+1]}\tilde{d}^{(i)}_jp^*_j\leq p^*_{\numitems+1}\sum_{i\in\numbuyers}\frac{Z}{\numbuyers}+\sum_{j\in [\numitems]}p^*_j= p^*.(1,1,\ldots,1,Z).
 \end{equation*}
As a result, $p^*.(d^*-(1,1,\ldots,1,Z))\leq 0$. By the definition of $\mathcal{F}(p^*)$ and due to the fact that at all-zero price vector the dot product is zero, we conclude $p^*.(d^*-(1,1,\ldots,1,Z))= 0$. If $p^*_{\numitems+1}=0$ then the seller in the market, i.e. agent $\numbuyers+1$, will buy infinite amount of money item due to lack of a budget constraint, and hence $d^*_{\numitems+1}>Z$. By definition of $\mathcal{F}(p^*)$, this implies $p^*.(d^*-(1,1,\ldots,1,Z))>0$, a contradiction. So $p^*_{\numitems+1}>0$.  Now if $d^*_j>1$ for some $j\in [\numitems]$, again by definition of $\mathcal{F}(p^*)$,  $p^*.(d^*-(1,1,\ldots,1,Z))>0$ which is a contradiction. So, $d^*_{j}\leq 1$ for all $j\in[\numitems]$. Because utilities are strictly increasing with money in the market, $d^*_{\numitems+1}=Z$. 
 Note that $\tilde{d}^{(\numbuyers+1)}_j=0$ for $j\in[\numitems]$, because seller is not interested in buying back any items. Therefore:
\begin{equation}
\label{eq:fracWEfeas2}
\forall j\in[\numitems]:~~\sum_{i\in[\numbuyers], S:j\in S}x_{i,S}=\sum_{i\in[\numbuyers]}\tilde{d}^{(i)}_j=d^*_{j}\leq 1.
\end{equation}

Finally, let $p_j\triangleq\frac{p^*_j}{p^*_{\numitems+1}}$ (this is possible because $p^*_{\numitems+1}>0$). First of all, we know that if $p_j>0$, then $p^*_j>0$. This implies $d^*_j=1$, as $p^*.(d^*-(1,1,\ldots,Z))=0$ and $\forall j'\in[\numitems]:d^*_{j'}\leq 1$. As a result
\begin{equation}
\label{eq:fracWEclear}
 p_j>0: \sum_{i\in[\numbuyers], S:j\in S}x_{i,S}=d^*_{j}=1
 \end{equation}
We know if $x_{i,S}>0$, then $(\mathbf{1}_S,\alpha_S)\in\D_i(p^*)$ for some $\alpha_S\geq 0$. Similar to the proof of Lemma~\ref{lem:reduction}, $\alpha_S$ is such that $p^*_{\numitems+1}\alpha_S+\sum_{j\in[\numitems]}\mathbf{1}_S(j)p^*_j=\frac{Z}{\numbuyers}$. Hence, similar to what happens in that proof, we can easily verify that $\mathbf{1}_S$ will be a demanded allocation for agent $i$ under prices $p^*$ in the combinatorial auction, as $(\mathbf{1}_S,\alpha_S)$ is a demanded allocation of goods for agent $i$ in its corresponding market. I.e. 
\begin{equation}
\label{eq:fracWEsat}
\mathbf{1}_S\in\underset{x'\in \binaryitems}{\argmax}u_i(x',\sum_{j=1}^{\numitems} p_jx'_j)
\end{equation}

So, $(\{x_{i,S}\}_{i\in\buyers,S\in 2^{\items}},\{p_j\}_{j\in\items})$ will be a fractional WE as in Definition~\ref{def:fracWE}, due to Equations~\ref{eq:fracWEfeas}, \ref{eq:fracWEfeas2}, \ref{eq:fracWEclear}, and \ref{eq:fracWEsat}.\qed

\subsection{Proof of Theorem~\ref{thm:intWE}.} 

\subsubsection{[Part 1, `if' direction]}
Suppose $(\tilde{p},\tilde{d})$ is a fixed point of the correspondence $\phi$ (this fixed point always exists, due to Lemma~\ref{lem:fixed}, and $\tilde{p}_{\numgoods}>0$). Now, following the proof of Theorem~\ref{thm:fracWE}, there exists a fractional WE $(\{\tilde{x}_{i,S}\}_{i\in\buyers,S\subseteq\items},\{p_j\}_{j\in\items})$, as in Definition~\ref{def:fracWE}, such that  $j\in[\numitems]: p_j=\frac{\tilde{p}_j}{\tilde{p}_\numgoods}$. Now fix $p^*=p$ and consider the induced configuration LP at $p^*$. Note that $\{\tilde{x}_{i,S}\}$ is a feasible solution for this LP by the definition of fractional WE. For every $i\in\buyers$ let $u_i=\underset{S\subseteq\items}{\max}~ u_i(\mathbf{1}_S,\sum_{j\in S}p_j)$. Then $(u,p)$ will form a feasible solution for the dual of induced configuration LP at price $p^*=p$, simply because $\forall i\in\buyers: u_i\geq 0$, $\forall j\in\items: p_j\geq 0$ and we have:
\begin{equation}
\forall i\in B, S\subseteq\items: u_i+\sum_{j\in S}p_j\geq u_i(\mathbf{1}_S,\sum_{j\in S}p_j)+\sum_{j\in S}p_j=v^{(p^*)}_i(\mathbf{1}_S)
\end{equation}

We next prove that $(u,p)$ will form an optimal solution for the dual of induced configuration LP at price $p^*=p$, by supporting this feasible dual solution with the feasible primal solution $\{\tilde{x}_{i,S}\}$. This can be done by the method of complementary slackness as following:
\begin{itemize}
\item if $\tilde{x}_{i,S}>0$, then by Definition~\ref{def:fracWE} we have $\mathbf{1}_S\in\underset{x'\in \binaryitems}{\argmax}u_i(x',\sum_{j=1}^{\numitems} p_jx'_j)$. Therefore, 
\begin{equation}
u_i+\sum_{j\in S}p_j=u_i(\mathbf{1}_S,\sum_{j\in S}p_j)+\sum_{j\in S}p_j=v^{(p^*)}_i(\mathbf{1}_S).
\end{equation}
\item if $p_j>0$, then by Definition~\ref{def:fracWE} we have $\sum_{i\in\buyers,S:j\in S}\tilde{x}_{i,S}=\tilde{d}_j=1$.
\item Due to the proof of Theorem~\ref{thm:fracWE}, there always exists at least one point in the convex combination of demanded vectors of buyer $i$, and hence if $u_i>0$ then $\sum_{S\subseteq \items} x_{i,S}=1$.
\end{itemize}
So $(u,p)$ is an optimal dual solution. Now suppose $\{x_{i,S}\}$ be an optimal integral solution to configuration LP at price $p$. Accordingly, $\{x_{i,S}\}$ and $  (u,p)$ should satisfy complementary slackness conditions. These conditions show why $(\{x_{i,S}\},p)$ will form a WE:
\begin{itemize}
\item Proof of satisfaction: due to complementary slackness, if buyer $i$ gets bundles $S$, then $x_{i,S}=1>0$, and therefore:
\begin{equation}
u_i+\sum_{j\in S}p_j=v^{(p^*)}_i(\mathbf{1}_S)\Rightarrow u_i(\mathbf{1}_S,\sum_{j\in S}p_j)=u_i=\underset{S'\subseteq\items}{\max}~u_i(\mathbf{1}_{S'},\sum_{j\in S'}p_j).
\end{equation}
\item Proof of market clearance: due to complementary slackness, if $p_j>0$, then $\displaystyle\sum_{i\in\buyers,S:j\in S}x_{i,S}=1$, as desired.
\end{itemize}
So, putting all pieces together, we conclude that $(x,p)$ is a WE for the combinatorial auction, where $p_j=\frac{\tilde{p}_j}{\tilde{p}_{\numgoods}}$, $(\tilde{p},\tilde{d})$ is a fixed point of the market correspondence $\phi$, as stated in the Theorem~\ref{thm:intWE}, and $\xxi_j=1\iff j\in S,~ x_{i,S}=1$\qed

\subsubsection{[Part 2, `only if' direction]}
Suppose $(\{x_{i,S}\}_{i\in\buyers,S\subseteq\items},\{p_j\}_{j\in\items})$ forms a WE. It is easy to see that $\{x_{i,S}\}$ will be an optimal solution to the induced configuration LP at price $p^*=p$. To see this, let $u_{i}=\underset{S\subseteq\items}{\max}~ u_i(\mathbf{1}_S,\sum_{j\in S}p_j)$. Similar to the proof of Part 1, $(u,p)$ forms a feasible solution for the dual LP of the induced configuration LP at price $p^*=p$. Now, we claim $\{x_{i,S}\}$ together with $(u,p)$ satisfy complementary slackness conditions, which shows $\{x_{i,S}\}$ and $(u,p)$ are optimal solutions to the primal and dual of the induced  configuration LP at price $p^*=p$ respectively. To see this, we have the followings:
\begin{itemize}
\item If $x_{i,S}>0$, then satisfaction property implies that $\mathbf{1}_S\in\underset{x'\in \binaryitems}{\argmax}u_i(x',\sum_{j=1}^{\numitems} p_jx'_j)$, and therefore: 
\begin{equation}
u_i+\sum_{j\in S}p_j=u_i(\mathbf{1}_S,\sum_{j\in S}p_j)+\sum_{j\in S}p_j=v^{(p^*)}_i(\mathbf{1}_S).
\end{equation}
\item If $p_j>0$, then market clearance property implies $\sum_{i\in\buyers,S:j\in S}{x}_{i,S}=1$.
\item If $u_i>0$, then this implies bidder $i$ gets a non-empty bundle of items under this WE, and hence $\sum_{S\subseteq \items} x_{i,S}=1$.
\end{itemize}
which is as desired. Now, for every $j\in\items$, let $\tilde{d}_j\triangleq\sum_{i\in\buyers,S:j\in S}{x}_{i,S}$ and $\tilde{p}_j\triangleq p_j$. Also, let $\tilde{d}_{\numgoods}=Z$ and $\tilde{p}_{\numgoods}=1$. We next show that $(\tilde{p},\tilde{d})$ is a fixed point of the market correspondence map $\phi$. To this end, for every $i\in\buyers$ let $S_i$ denotes the bundle of items such that $x_{i,S_i}=1$. Moreover, let $d^{(i)}_{\numgoods}=\frac{Z}{\numbuyers}-\sum_{j\in S_i}\tilde{p}_j$. We show $(\mathbf{1}_{S_i},d^{(i)}_{\numgoods})\in D_i(\tilde{p})$. Note that $\tilde{p}_{\numgoods}d^{(i)}_{\numgoods}+\sum_{j\in [\numgoods-1]}\mathbf{1}_{S_i}(j)\tilde{p}_j=\tilde{p}_{\numgoods}\frac{Z}{n}$. Moreover, if $x'\in\{0,1\}^{\numgoods-1}\times\mathbb{R}_+$ satisfies the budget constraint of agent $i$, i.e. $x'_{\numgoods}\tilde{p}_{\numgoods}+\sum_{j\in[\numgoods-1]}x'_j\tilde{p}_j\leq \tilde{p}_{\numgoods}\frac{Z}{\numbuyers}$, the following holds:
\begin{align*}
\tilde{u}_i(\mathbf{1}_{S_i},d^{(i)}_{\numgoods})&=u_i(\mathbf{1}_{S_i},\frac{Z}{\numbuyers}-d^{(i)}_{\numgoods})=u_i(\mathbf{1}_{S_i},\sum_{j\in S_i}\tilde{p}_j)\\
&\overset{(1)}{\geq} u_i(\{x'_j\}_{j\in [\numitems]},\sum_{j\in [\numitems]}\tilde{p}_jx'_j)\overset{(2)}{\geq} u_i(\{x'_j\}_{j\in [\numitems]},\frac{Z}{\numbuyers}-x'_{\numgoods})=\tilde{u}_i(x'),
\end{align*}
in which (1) holds because $S_i$ is an optimal bundle for bidder $i$ under prices $\{\tilde{p}\}_{j\in\items}$ and (2) holds because utilities $\{u_i(.)\}$ are decreasing with respect to money in the combinatorial auction. Therefore $d^{(i)}=(\mathbf{1}_{S_i},d^{(i)}_{\numgoods})\in D_i(\tilde{p})$ for $i\in[\numbuyers]$. Moreover, it is clear that $d^{(\numagents)}=(\mathbf{0},\sum_{j=1}^{\numitems}\tilde{p}_j)\in D_{\numagents}(\tilde{p})$. So, $\tilde{d}=\sum_{i\in\agents}d^{(i)}\in D(\tilde{p})\subseteq \tilde{D}(\tilde{p})$. Furthermore, $\tilde{d}_j\leq 1$ for $j\in[\numitems]$, and $\tilde{d}_{\numgoods}=Z$. Hence, for every price vector $p\in\mathbb{R}_+^{\numgoods}$, $p.(\tilde{d}-(1,1,\ldots,1,Z))\leq 0$. Due to the market clearance of the WE, we know that if $\tilde{p}_j>0$, then $\tilde{d}_j=1$. As a result, $\tilde{p}.(\tilde{d}-(1,1,\ldots,1,Z))=0$. So, $\tilde{p}\in\mathcal{F}(\tilde{d})$. Putting pieces together, $\tilde{p}\in\mathcal{F}(\tilde{d})$ and $\tilde{d}\in\tilde{D}(\tilde{p})$, and hence $(\tilde{p},\tilde{d})$ is a fixed-point of the market correspondence $\phi$.\qed

\section{Conclusion}

In the study of Walrasian equilibria, it is standard to assume that bidders have utilities that are quasilinear in money. Unfortunately, this is a strong assumption that has attracted little attention. We strive to study Walrasian equilibria in general combinatorial auction settings without assuming utilities are quasilinear, and our main results shed light on when they exist. Unsurprisingly, we find that some of the strong results for quasilinear bidders break when we relax our utility model. We show structure that does exist, and how it connects to a few key results for general quasilinear and unit demand non-quasilinear settings; however, we have only touched a small fraction of what is known about the quasilinear setting, and that is one source of interesting open questions. For example:
\begin{itemize}
\item {\em What natural properties of utility functions guarantee the existence of Walrasian equilibria?} In quasilinear settings, it is known that gross substitutes is sufficient in a combinatorial auction.
\item {\em When do equilibria have a lattice structure?} It is known that in quasilinear settings and in unit-demand non-quasiliniear ones, Walrasian equilibria have a lattice structure. Does this exist more generally?
\end{itemize}
Another direction for research surrounds generalizations of Walrasian equilibria:
\begin{itemize}
\item {\em When do combinatorial Walrasian equilibria exist in general?} In quasilinear settings where Walrasian equilibria fail to exist, one line of research shows that a kind of combinatorial equilibrium does exist~\cite{feldman2016combinatorial}.
\end{itemize}
A third direction for research is to ask when quasilinearity is justified:
\begin{itemize}
\item {\em What kinds of micromarkets naturally lead to quasilinear relationships with the global market?} Taking the view that Walrasian equilibria capture a small slice of a market, one should be able to identify conditions under which a micromarket naturally has a quasilinear relationship with the other options available to an agent.
\end{itemize}
\appendix
\bibliography{main}

\begin{thebibliography}{}

\bibitem[\protect\citeauthoryear{Alaei \bgroup \em et al.\egroup
  }{2011}]{alaei2011competitive}
Saeed Alaei, Kamal Jain, and Azarakhsh Malekian.
\newblock Competitive equilibrium in two sided matching markets with general
  utility functions.
\newblock {\em ACM SIGecom Exchanges}, 10(2):34--36, 2011.

\bibitem[\protect\citeauthoryear{Arrow and Debreu}{1954}]{arrow1954existence}
Kenneth~J Arrow and Gerard Debreu.
\newblock Existence of an equilibrium for a competitive economy.
\newblock {\em Econometrica: Journal of the Econometric Society}, pages
  265--290, 1954.

\bibitem[\protect\citeauthoryear{Bevia \bgroup \em et al.\egroup
  }{1999}]{bevia1999buying}
Carmen Bevia, Martine Quinzii, and Jose~A Silva.
\newblock Buying several indivisible goods.
\newblock {\em Mathematical Social Sciences}, 37(1):1--23, 1999.

\bibitem[\protect\citeauthoryear{Bikhchandani and
  Mamer}{1997}]{bikhchandani1997competitive}
Sushil Bikhchandani and John~W Mamer.
\newblock Competitive equilibrium in an exchange economy with indivisibilities.
\newblock {\em Journal of economic theory}, 74(2):385--413, 1997.

\bibitem[\protect\citeauthoryear{Clarke}{1971}]{clarke1971multipart}
Edward~H Clarke.
\newblock Multipart pricing of public goods.
\newblock {\em Public choice}, 11(1):17--33, 1971.

\bibitem[\protect\citeauthoryear{Cramton \bgroup \em et al.\egroup
  }{2006}]{cramton2006combinatorial}
Peter~C Cramton, Yoav Shoham, Richard Steinberg, et~al.
\newblock {\em Combinatorial auctions}, volume 475.
\newblock MIT press Cambridge, 2006.

\bibitem[\protect\citeauthoryear{Demange and Gale}{1985}]{demange1985strategy}
Gabrielle Demange and David Gale.
\newblock The strategy structure of two-sided matching markets.
\newblock {\em Econometrica: Journal of the Econometric Society}, pages
  873--888, 1985.

\bibitem[\protect\citeauthoryear{Echenique and
  Oviedo}{2004}]{echenique2004theory}
Federico Echenique and Jorge Oviedo.
\newblock A theory of stability in many-to-many matching markets.
\newblock 2004.

\bibitem[\protect\citeauthoryear{Feldman \bgroup \em et al.\egroup
  }{2016}]{feldman2016combinatorial}
Michal Feldman, Nick Gravin, and Brendan Lucier.
\newblock Combinatorial walrasian equilibrium.
\newblock {\em SIAM Journal on Computing}, 45(1):29--48, 2016.

\bibitem[\protect\citeauthoryear{Gale}{1984}]{gale1984equilibrium}
David Gale.
\newblock Equilibrium in a discrete exchange economy with money.
\newblock {\em International Journal of Game Theory}, 13(1):61--64, 1984.

\bibitem[\protect\citeauthoryear{Groves}{1973}]{groves1973incentives}
Theodore Groves.
\newblock Incentives in teams.
\newblock {\em Econometrica: Journal of the Econometric Society}, pages
  617--631, 1973.

\bibitem[\protect\citeauthoryear{Gul and Stacchetti}{1999}]{gul1999walrasian}
Faruk Gul and Ennio Stacchetti.
\newblock Walrasian equilibrium with gross substitutes.
\newblock {\em Journal of Economic theory}, 87(1):95--124, 1999.

\bibitem[\protect\citeauthoryear{Kakutani and
  others}{1941}]{kakutani1941generalization}
Shizuo Kakutani et~al.
\newblock {\em A generalization of Brouwer's fixed point theorem}.
\newblock Duke University Press, 1941.

\bibitem[\protect\citeauthoryear{Kaneko and
  Yamamoto}{1986}]{kaneko1986existence}
Mamoru Kaneko and Yoshitsugu Yamamoto.
\newblock The existence and computation of competitive equilibria in markets
  with an indivisible commodity.
\newblock {\em Journal of Economic Theory}, 38(1):118--136, 1986.

\bibitem[\protect\citeauthoryear{Kelso~Jr and Crawford}{1982}]{kelso1982job}
Alexander~S Kelso~Jr and Vincent~P Crawford.
\newblock Job matching, coalition formation, and gross substitutes.
\newblock {\em Econometrica: Journal of the Econometric Society}, pages
  1483--1504, 1982.

\bibitem[\protect\citeauthoryear{Maskin}{1987}]{maskin1987fair}
Eric~S Maskin.
\newblock On the fair allocation of indivisible goods.
\newblock In {\em Arrow and the Foundations of the Theory of Economic Policy},
  pages 341--349. Springer, 1987.

\bibitem[\protect\citeauthoryear{Murota and
  Tamura}{2001}]{murota2001computation}
Kazuo Murota and Akihisa Tamura.
\newblock {\em Computation of Competitive Equilibria of Indivisible Commodities
  Via M-convex Submodular Flow Problem}.
\newblock Kyoto University, Research Institute for Mathematical Sciences, 2001.

\bibitem[\protect\citeauthoryear{Nisan \bgroup \em et al.\egroup
  }{2007}]{nisan2007algorithmic}
Noam Nisan, Tim Roughgarden, Eva Tardos, and Vijay~V Vazirani.
\newblock {\em Algorithmic game theory}, volume~1.
\newblock Cambridge University Press Cambridge, 2007.

\bibitem[\protect\citeauthoryear{Quinzii}{1984}]{quinzii1984core}
Martin Quinzii.
\newblock Core and competitive equilibria with indivisibilities.
\newblock {\em International Journal of Game Theory}, 13(1):41--60, 1984.

\bibitem[\protect\citeauthoryear{Schrijver}{1983}]{schrijver1983short}
Alexander Schrijver.
\newblock Short proofs on the matching polyhedron.
\newblock {\em Journal of Combinatorial Theory, Series B}, 34(1):104--108,
  1983.

\bibitem[\protect\citeauthoryear{Shapley and
  Shubik}{1971}]{shapley1971assignment}
Lloyd~S Shapley and Martin Shubik.
\newblock The assignment game i: The core.
\newblock {\em International Journal of game theory}, 1(1):111--130, 1971.

\bibitem[\protect\citeauthoryear{Svensson}{1984}]{svensson1984competitive}
Lars-Gunnar Svensson.
\newblock Competitive equilibria with indivisible goods.
\newblock {\em Journal of Economics}, 44(4):373--386, 1984.

\bibitem[\protect\citeauthoryear{Vickrey}{1961}]{vickrey1961counterspeculation}
William Vickrey.
\newblock Counterspeculation, auctions, and competitive sealed tenders.
\newblock {\em The Journal of finance}, 16(1):8--37, 1961.

\end{thebibliography}
\end{document}